\documentclass[titlepage]{article}
\pdfoutput=1
\usepackage{natbib}
\usepackage{authblk}

\usepackage{enumerate,amsmath,amssymb,amsthm}
\usepackage{epsfig}
\usepackage{graphicx}
\usepackage{fancyhdr}
\usepackage{ifpdf}
\usepackage{epstopdf}
\usepackage{setspace}
\usepackage{booktabs,rotating}

\usepackage{hypernat}
\usepackage{color}
\usepackage[colorlinks=true, urlcolor=citecolor, linkcolor=citecolor, citecolor=citecolor]{hyperref}
\hypersetup{linkcolor=blue,citecolor=blue,filecolor=black,urlcolor=MidnightBlue}

\usepackage{color}

\newtheorem{Corollary}{Corollary}
\newtheorem{Proposition}{Proposition}

\allowdisplaybreaks

\usepackage[pagewise]{lineno}
\linenumbers*[1]
\newcommand*\patchAmsMathEnvironmentForLineno[1]{%
  \expandafter\let\csname old#1\expandafter\endcsname\csname #1\endcsname
  \expandafter\let\csname oldend#1\expandafter\endcsname\csname end#1\endcsname
  \renewenvironment{#1}%
  {\linenomath\csname old#1\endcsname}%
  {\csname oldend#1\endcsname\endlinenomath}}%
\newcommand*\patchBothAmsMathEnvironmentsForLineno[1]{%
  \patchAmsMathEnvironmentForLineno{#1}%
  \patchAmsMathEnvironmentForLineno{#1*}}%
\AtBeginDocument{%
  \patchBothAmsMathEnvironmentsForLineno{equation}%
  \patchBothAmsMathEnvironmentsForLineno{align}%
  \patchBothAmsMathEnvironmentsForLineno{flalign}%
  \patchBothAmsMathEnvironmentsForLineno{alignat}%
  \patchBothAmsMathEnvironmentsForLineno{gather}%
  \patchBothAmsMathEnvironmentsForLineno{multline}%
}

\begin{document}

\def\spacingset#1{\renewcommand{\baselinestretch}%
  {#1}\small\normalsize} \spacingset{1}

% NOTE: To produce blinded version, replace "0" with "1" below.
\newcommand{\blind}{0}

\setcounter{footnote}{0}

\title{Density and Distribution Evaluation for Convolution of Independent Gamma Variables}

\if0\blind
{
  \author{
    Chaoran Hu, Vladimir Pozdnyakov, and Jun Yan\\
    Department of Statistics\\
    University of Connecticut\\
    215 Glenbrook Road, Storrs, Connecticut 06269, U.S.A.
    }
}\fi

\if1\blind
{
\author{}
} \fi

\maketitle

\begin{abstract} % 40 words limit
Several numerical evaluations of the density and distribution of convolution
of independent gamma variables are compared in their accuracy and speed.
In application to renewal processes, an efficient formula is derived for the
probability mass function of the event count.

\bigskip

\noindent%
\end{abstract}

\doublespacing

\section{Introduction}
\label{sec:intr}

Fast and precise evaluation of the density and distribution function of
convolution of independent gamma variables is important in many applications
such as storage capacity measurement \citep[e.g.,][]{mathai1982storage},
reliability analysis \citep[e.g.,][]{kadri2015markov}, and point processes
\citep[e.g.,][]{sim1992point}. Let $X_{1}, \dots, X_{n}$ be $n$ mutually
independent random variables that have
gamma distributions with shape parameters $\alpha_{i} > 0$ and scale
parameters $\beta_{i} > 0$, $j = 1, \dots, n$. Then random variable
$Y = \sum_{i=1}^n X_i$ is the convolution of independent gamma variables.
Without loss of generality, the scale parameters $\beta_i$'s can be assumed to
of distinctive values; otherwise, variables with the same scales can be summed
before the convolution. Evaluation of the density
and distribution of $Y$ is the focus of this paper.

Exact evaluations, which do not have simple closed forms, are challenging.
\citet{mathai1982storage} was the first to give an expression of the density
in terms of multiple infinite series for the general case of arbitrary shape
and scale parameters; the multiple series is a confluent hypergeometric
function of $n - 1$ variables. When $n = 2$, the confluent hypergeometric
function is univariate with efficient implementations in the GNU Scientific
Library (GSL) \citet{galassi2009gnu}, which was adopted by \citet{di2006exact}.
\citet{mathai1982storage} also gave relatively easier-to-compute expressions
in the special cases where all shapes are integers or are identical.
\citet{moschopoulos1985distribution} simplified the complex expression into a
single gamma-series representation along with a formula for the truncation
error to evaluate the precision of numerical computation.
\citet{akkouchi2005convolution} indicated that the density of $Y$
can be expressed through an integral of the generalized beta function, but did not
give explicit ways to numerically evaluate the integral. More recently,
\citet{vellaisamy2009sums} proposed a random parameter representation of the
gamma-series of \citet{moschopoulos1985distribution}, where the weights define
the probability mass function of a discrete distribution on non-negative
integers. Implementations of the gamma-series methods are simple but the
computation is too much CPU-time consuming when the variability of the scale
parameters is large and the shape parameters are small.
Built on the representation of \citet{vellaisamy2009sums},
\citet{barnabani2017approximation} proposed a fast approximation which
approximates the weights of the gamma-series by a discrete distribution.

The contribution of this article is two-fold.
First, we give a computationally review of the methods of
\citet{mathai1982storage} and \citet{moschopoulos1985distribution} and provide
their implementations in our R package \texttt{coga} \citep{Rpkg:coga}.
Their speeds are compared in cases of $n = 2$ and $n = 3$.
In the case of $n = 3$, the accuracy of the fast approximation of
\citet{barnabani2017approximation} is also assessed using our implementation.
Second, in an application to renewal processes with holding times following a
mixture of exponential distributions, we derive a new formula for the
probability mass function of the number of renewals by a given amount of time,
which provides very fast exact evaluations in a numerical study.

\section{Exact Evaluations}
\label{sec:exact}

Let us introduce some notations first. Let $G(y; \alpha, \beta)$ and
$g(y; \alpha, \beta)$ be, respectively, the distribution and density function
of a gamma variable with shape $\alpha$, and scale $\beta$. For the special
case where $\beta = 1$, we use $G(y; \alpha)$ and $g(y; \alpha)$,
respectively. Let $F(x; (\alpha_1, \beta_1), (\alpha_2, \beta_2))$ and
$f(x; (\alpha_1, \beta_1), (\alpha_2, \beta_2))$, respectively, be the
distribution and density function of $Y$ in the case of $n = 2$.
Finally, let $(x)_m = x(x+1) \dots (x+m-1)$ be the Pochhammer polynomial
\citep[Equation~6.1.22]{abramowitz1972handbook}.

\subsection{Mathai's Method}
\label{subsec:mat}

\citet{mathai1982storage} expresses the density of $Y$ via a multiple infinite
series
\begin{equation}
  \label{eq:den:mat}
  f(x) = \left[\prod_{j=1}^{n}\beta_j^{\alpha_j}\Gamma(\gamma)\right]^{-1}
  x^{\gamma-1}e^{-x/\beta_1}
 \phi\big(
  \alpha_2, \ldots, \alpha_n; \gamma;
  (1/\beta_1 - 1/\beta_2) x, \ldots, (1/\beta_1 - 1/\beta_n) x
  \big),
\end{equation}
where $\beta_1 = \min_j(\beta_j)$, $\gamma = \sum_{j=1}^{n} \alpha_j$,
and $\phi$ is a confluent hypergeometric
function of $n - 1 $ variables defined by a multiple series
\begin{align*}
  &\, \phi\big(
    \alpha_2, \ldots, \alpha_n; \gamma;
    (1/\beta_1 - 1/\beta_2) x, \ldots, (1/\beta_1 - 1/\beta_n) x
    \big)\\
  = &\sum_{r_2=0}^{\infty} \dots \sum_{r_n=0}^{\infty} \Big\{ (\alpha_2)_{r_2}
    \dots (\alpha_n)_{r_n}
    \left[(1/\beta_1-1/\beta_2)x \right]^{r_2} \dots
      \left[(1/\beta_1-1/\beta_n)x \right]^{r_n} / \left[r_2! \dots r_n! (\gamma)_r \right] \Big\},
\end{align*}
a special function which has been studied in the literature
\citep{mathai1978hfunction}.
With the gamma function kernels, $x^{\gamma + r_i - 1}e^{-x/\beta_1}$,
The distribution function can be expressed in terms of incomplete gamma
functions by term-by-term integration of Equation~\eqref{eq:den:mat}.

For the special case of $n = 2$, the density is expressed
in terms of the Kummer confluent hypergeometric function, $_1F_1$ as
\citep[Formula~13.1.2]{abramowitz1972handbook},
\begin{equation}
  \label{eq:dco2ga:hyper}
  \begin{aligned}
  f(x; (\alpha_1, \beta_1), (\alpha_2, \beta_2)) &= \frac{x^{\gamma - 1} e^{- x / \beta_1}}
  {\beta_1^{\alpha_1} \beta_2^{\alpha_2} \Gamma(\gamma)}
  {}_1F_1(\alpha_2; \gamma; (1/\beta_{1} - 1/\beta_{2})x)\\
  &= \left( \frac{\beta_1}{\beta_2} \right)^{\alpha_2} g(x; \gamma, \beta_1)
  {}_1F_1(\alpha_2; \gamma; (1/\beta_{1} - 1/\beta_{2})x).
  \end{aligned}
\end{equation}
The benefit of Equation~\eqref{eq:dco2ga:hyper} is that the GSL
\citep{galassi2009gnu} has an implementation of $_1F_1$. Note that the
condition $\beta_1 < \beta_{2}$ is not needed in~\eqref{eq:dco2ga:hyper}.
Indeed, if $\beta_1 > \beta_2$, then
\begin{align*}
    f(x; (\alpha_1, \beta_1), (\alpha_2, \beta_2))
    &= f(x; (\alpha_2, \beta_2), (\alpha_1, \beta_1))\\
    &= \frac{x^{\gamma - 1} e^{- x / \beta_2}}
    {\beta_1^{\alpha_1} \beta_2^{\alpha_2} \Gamma(\gamma)}
    {}_1F_1(\alpha_1; \gamma; (1/\beta_{2} - 1/\beta_{1})x)\\
    &= \frac{x^{\gamma - 1} e^{- x / \beta_2}}
    {\beta_1^{\alpha_1} \beta_2^{\alpha_2} \Gamma(\gamma)}
    e^{-(1/\beta_2 - 1/\beta_1)x}
    {}_1F_1(\alpha_2; \gamma; (1/\beta_{1} - 1/\beta_{2})x)\\
    &= \frac{x^{\gamma - 1} e^{- x / \beta_1}}
    {\beta_1^{\alpha_1} \beta_2^{\alpha_2} \Gamma(\gamma)}
    {}_1F_1(\alpha_2; \gamma; (1/\beta_{1} - 1/\beta_{2})x),
\end{align*}
where the third equation follows from ${}_1F_1(a; b; z) \equiv e^z {}_1F_1(b -
a; b; -z)$ \citep[Equation~13.1.27]{abramowitz1972handbook}.

The distribution function when $n = 2$ can be explicitly expressed as,
for $y > 0$,
\begin{equation}
  \begin{aligned} \label{eq:pco2ga}
   & F(y; (\alpha_1, \beta_1), (\alpha_2, \beta_2))\\
    =\, & \frac{1}{\beta_1^{\alpha_1} \beta_2^{\alpha_2}}
    \sum_{k = 0}^{\infty} \frac{\binom{\alpha_2 + k - 1}{k}}{\Gamma(\gamma + k)}
    (1/\beta_1 - 1/\beta_2)^k
    \int_{0}^{y} x^{k + \gamma - 1} e^{-x/\beta_1} dx\\
    =\,& \frac{1}{\beta_1^{\alpha_1} \beta_2^{\alpha_2}}
    \sum_{k = 0}^{\infty} \frac{\binom{\alpha_2 + k - 1}{k}}{\Gamma(\gamma + k)}
    (1/\beta_1 - 1/\beta_2)^k
    \beta_1^{k+\gamma} G(y/\beta_1; k+\gamma) \Gamma(k+\gamma)\\
    =\,& \left(\frac{\beta_1}{\beta_2}\right)^{\alpha_2}
    \sum_{k = 0}^{\infty} \binom{\alpha_2 + k - 1}{k} (1-\beta_1/\beta_2)^k
    G(y/\beta_1; k + \gamma).
  \end{aligned}
\end{equation}

\subsection{Moschopoulos' Method}
\label{subsec:mos}

\citet{moschopoulos1985distribution} expresses the density of $Y$
by a single gamma series with coefficients that can be calculated recursively:
\begin{align*}
  f(x)
  =\, &C \sum_{k = 0}^{\infty} \delta_{k} x^{\rho + k - 1} e^{-x/\beta_1}
   / \left[\Gamma(\rho + k) \beta_1^{\rho + k} \right]\\
  =\, & C \sum_{k = 0}^{\infty} \delta_{k} g(x; \rho + k, \beta_1), x > 0,
\end{align*}
where $\beta_1=\min_i(\beta_i)$, $C=\prod_{i = 1}^{n}(\beta_1/\beta_i)^{\alpha_i}$,
$\rho=\sum_{i=1}^{n}\alpha_i > 0$, and $\delta_k$ is given by the recursive
relations
\begin{equation*}
  \label{eq:rec:mos}
  \delta_{k+1} = \frac{1}{k+1}\sum_{i=1}^{k+1} i \gamma_i \delta_{k+1-i},
  k =0,1,2,\dots,
\end{equation*}
with $\delta_0=1$ and $\gamma_k = \sum_{i=1}^{n} \alpha_i(1-\beta_1/\beta_i)^k
/k$ for $k = 1,2,\dots$. This expression facilitates distribution function
evaluation as
\begin{equation*}
  F(y) = C \sum_{k=0}^{\infty} \delta_k G(y; \rho + k, \beta_1), \qquad y > 0.
\end{equation*}
The weights $C\delta_k$'s can be viewed as the probability masses of a
discrete random variable on non-negative integers \citep{vellaisamy2009sums}.
When $n = 2$, this discrete distribution is negative binomial. For $n > 2$,
\citet{barnabani2017approximation} proposed to approximate the discrete
distribution by a three-parameter generalized negative binomial distribution
defined by \citet{jain1971generalized} through moment matching.

\subsection{Timing Comparison}
\label{subsec:timeinf}

We implemented the methods of \citet{mathai1982storage} and
\citet{moschopoulos1985distribution} in an open source R package
\texttt{coga} \citep{Rpkg:coga}. The computation is done in \texttt{C++} code
and interfaced to R \citep{R} in the \texttt{coga} package.
In addition, the fast approximation of \citet{barnabani2017approximation} is
also available in the package for $n > 2$.

\begin{table}[tbp]
  \centering
  \caption{Timing comparison (in microseconds) of Mathai's and Moschopoulos' methods
    when $n = 2$ in evaluating the density and distribution function of
    convolutions of independent gamma variables.}
  \label{tab:coga2}
\begin{tabular}{rrrrrrr}
  \toprule
 \multicolumn{3}{c}{Parameters}
 & \multicolumn{2}{c}{Density} & \multicolumn{2}{c}{Distribution Function}\\
 \cmidrule(r){1-3} \cmidrule(r){4-5} \cmidrule(r){6-7}
 $\alpha$ & $\beta_1$ & $\beta_2$ & Moschopoulos & Mathai & Moschopoulos & Mathai \\
  \midrule
  0.2 & 0.4 & 0.3 & 23,030 & 86 & 25,011 & 2,648 \\ 
   & 4 & 0.3 & 103,987 & 176 & 110,804 & 9,018 \\ 
   & 4 & 3 & 24,566 & 88 & 25,696 & 2,849 \\ 
  2 & 0.4 & 0.3 & 29,163 & 94 & 31,030 & 3,086 \\ 
   & 4 & 0.3 & 165,901 & 96 & 173,916 & 13,590 \\ 
   & 4 & 3 & 30,378 & 90 & 33,489 & 3,397 \\ 
  20 & 0.4 & 0.3 & 53,231 & 103 & 57,468 & 6,471 \\ 
   & 4 & 0.3 & 538,807 & 175 & 566,857 & 37,527 \\ 
   & 4 & 3 & 53,259 & 108 & 58,479 & 6,604 \\ 
 \bottomrule
\end{tabular}
\end{table}

We first compare the speed of the two methods in the case of $n = 2$.
% Since the gamma-series of \citet{moschopoulos1985distribution} is CPU-time
% intensive when the shape parameters are small and the scale parameters have
% large variation \citep{barnabani2017approximation}, we considered such
% situations and its complements.
The shape parameters were set to be
$\alpha_1 =  \alpha_2 \in \{0.2, 2, 20\}$.
The scale parameters were set to be
$(\beta_1, \beta_2) \in \{(0.4, 0.3), (4, 0.3), (4, 3) \}$.
For each configuration, we used a large number ($100,000$) of
simulated observations from the distribution to determine the bulk
range of the observations. Then we evaluate the density and the distribution
of the convolution over 100~equally spaced grid points in the bulk range.
The evaluations were repeated 100 times.

Table~\ref{tab:coga2} summarizes the median time to evaluate density and
distribution at the 100-point grid from 100 replicates obtained on an Intel
2.50GHz computer. Density evaluation using Mathai's method (implemented with
the $_1F_1$ function from the GSL) performs much faster than Moschopoulos'
method in all settings; in some settings, it is up to 3,000 times faster
Distribution evaluation takes much longer than density evaluation using
Mathai's method, but is still up to 16 times faster than Moschopoulos' method.
Moschopoulos' method takes longer when the scale parameters are very
different, $(\beta_1, \beta_2) = (4, 0.3)$, or the shape parameters bigger.
Mathai's method is much less sensitive to the parameter settings.

\begin{table}[tbp]
  \centering
  \caption{Timing comparison (in milliseconds) of Mathai's, Moschopoulos'
    exact methods and Barnabani's approximation method when $n = 3$.}
  \label{tab:coga3}
\begin{tabular}{rrrrrrrrrr}
  \toprule
 \multicolumn{4}{c}{Parameters}
 & \multicolumn{3}{c}{Density}
 & \multicolumn{3}{c}{Distribution Function}\\
 \cmidrule(r){1-4} \cmidrule(r){5-7} \cmidrule(r){8-10}
 $\alpha$ & $\beta_1$ & $\beta_2$ & $\beta_3$ & Mosch. & Mathai & Approx. & Mosch. & Mathai & Approx.\\
  \midrule
  0.2 & 0.4 & 0.3 & 0.2 & 37 & 1,223 & 6 & 39 & 1,428 & 8 \\ 
   & 4 & 0.3 & 0.2 & 167 & 5,796 & 18 & 181 & 8,067 & 28 \\ 
   & 4 & 3 & 0.2 & 186 & 10,208 & 19 & 197 & 14,000 & 30 \\ 
   & 4 & 3 & 2 & 38 & 1,245 & 6 & 40 & 1,474 & 8 \\ 
  2 & 0.4 & 0.3 & 0.2 & 48 & 1,044 & 8 & 51 & 1,328 & 11 \\ 
   & 4 & 0.3 & 0.2 & 242 & 5,333 & 21 & 253 & 8,975 & 35 \\ 
   & 4 & 3 & 0.2 & 313 & 12,597 & 25 & 331 & 20,646 & 43 \\ 
   & 4 & 3 & 2 & 49 & 1,082 & 8 & 53 & 1,415 & 11 \\ 
  20 & 0.4 & 0.3 & 0.2 & 109 & 3,250 & 14 & 119 & 4,721 & 22 \\ 
   & 4 & 0.3 & 0.2 & 780 & 16,083 & 29 & 950 & 40,704 & 78 \\ 
   & 4 & 3 & 0.2 & 596 & 21,553 & 18 & 1,418 & 133,699 & 101 \\ 
   & 4 & 3 & 2 & 110 & 3,329 & 14 & 123 & 4,953 & 23 \\ 
  \bottomrule
\end{tabular}
\end{table}

Following the design and steps in the case of $n = 2$, we conducted a
numerical analysis for $n = 3$. The shape parameters were set to be
$\alpha_1 = \alpha_2 = \alpha_3 \in \{0.2, 2, 20\}$.
The scale parameters were set to be
$(\beta_1, \beta_2, \beta_3) \in \{(0.4, 0.3, 0.2), (4, 0.3, 0.2),
(4, 3, 0.2), (4, 3, 2)\}$.
The two methods gave numerically indistinguishable results in
both density and distribution evaluations, verifying each other.
Table~\ref{tab:coga3} summarize the median timing results from 100 replicates.
When $n > 2$, the multivariable confluent hypergeometric function has no efficient
implementations yet to the best of our knowledge. Therefore, Mathai's method
needs to evaluate $n - 1$ nested infinite series, which makes it quite
complicated \citep{jasiulewicz2003convolutions, sen1999convolution}.
In this case, Moschopoulos' method is preferred for its single-series.
The approximation of \citet{barnabani2017approximation} was also included in
the comparison, which is up to $30$ times faster in density evaluation and $14$
times faster in distribution evaluation than Moschopoulos' method.

\begin{figure}[tbp]
\centering
\includegraphics[angle=0, scale=0.8]{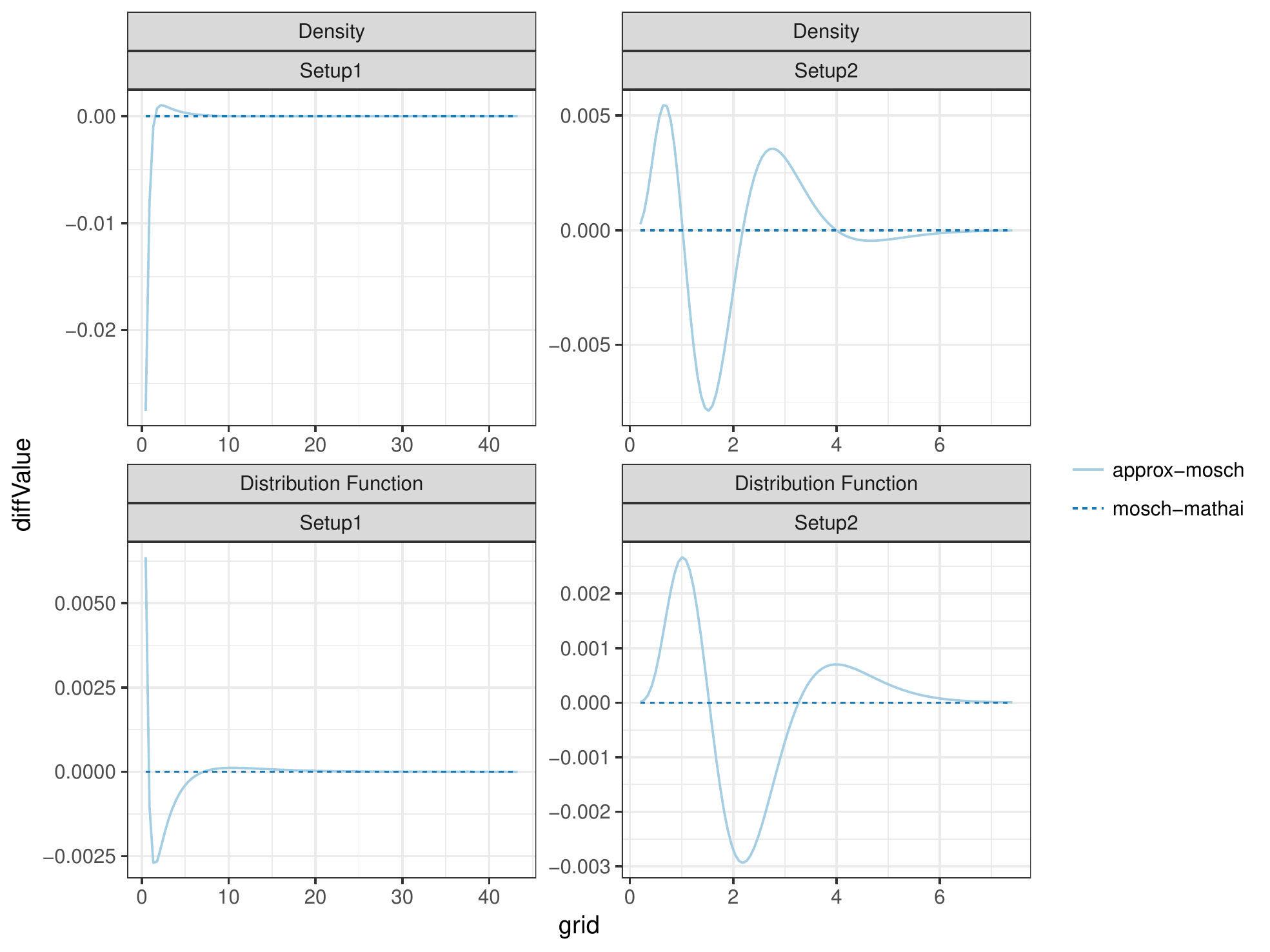}
\caption{Differences between the approximation method and the exact methods in
  evaluating the density and distribution function of convolution of three
  independent gamma variables.
  The parameter in setup~1 are $\alpha_1 = \alpha_2 = \alpha_3 = 0.2$,
  $\beta_1 = 4$, $\beta_2 = 0.3$, and $\beta_3 = 0.2$.
  The parameter in setup~2 are $\alpha_1 = \alpha_2 = \alpha_3 = 2$,
  $\beta_1 = 0.4$, $\beta_2 = 0.3$, and $\beta_3 = 0.2$.}
\label{fig:error}
\end{figure}

The exact implementations make it possible to assess the accuracy of the
approximation approach of \citet{barnabani2017approximation}.
The differences between the approximation and exact evaluation in the two
worst cases out of a collection of settings we experimented are shown in
Figure~\ref{fig:error}. Apparently, the approximation is very accurate.

\section{Application to a Renewal Process}
\label{sec:appl}

Let $\{ U_k \}_{k \geq 1}$ be independent and identically distributed random
variables following a mixture of two exponential distributions,
$\mathcal{E}xp(\beta_1)$ and $\mathcal{E}xp(\beta_2)$,
with weights $p \in (0, 1)$ and $1 - p$, respectively.
For any $t > 0$, define
\begin{equation}
  \label{eq:renew}
  N(t) = \sup\{ n \geq 0 : \sum_{k=1}^{n} U_k \leq t \},
\end{equation}
where, by convention, the summation over an empty set is~$0$. The process
$N(t)$, $t > 0$, is a renewal process. The following result gives an expression
of the distribution of $N(t)$ in terms of the Kummer confluent hypergeometric
function, which allows a very efficient numerical evaluation.

\begin{Proposition}\label{thm:renew}
For the renewal process defined in~\eqref{eq:renew} and integer $n \ge 0$,
\begin{align*}
  \Pr\big(N(t) = n\big) =
  \sum_{k=0}^{n} & \big\{p \psi(\beta_1, \beta_2, t, k, n) +
  (1 - p) \psi(\beta_2, \beta_1, t, n - k, n) \big\} \times \\
  & \frac{t^n}{\beta_1^k \beta_2^{n-k} \Gamma(n+1)} {n\choose k} p^k (1-p)^{n-k},
\end{align*}
where
\[
\psi(\beta_1, \beta_2, t, k, n) =
e^{-t/\beta_1} {}_1F_1\big(n-k;n+1;t(1/\beta_1-1/\beta_2)\big).
\]
\end{Proposition}

\begin{proof}
First, note that
\begin{equation*}
  \begin{aligned}
    &\Pr(N(t) = n) = \Pr \left(\sum_{k=1}^{n} U_k \leq t,
      \sum_{k=1}^{n+1} U_k > t \right)\\
    = & p  \Pr \left(\sum_{k=1}^{n} U_k \leq t, \sum_{k=1}^{n} U_k + E_1 > t\right)
    +(1-p) \Pr \left(\sum_{k=1}^{n} U_k \leq t, \sum_{k=1}^{n} U_k + E_2 > t\right),
  \end{aligned}
\end{equation*}
where $E_1$ and $E_2$ are independent of $\{ U_k \}_{k \geq 1}$ with
$\mathcal{E}xp(\beta_1)$ and $\mathcal{E}xp(\beta_2)$ distributions,
respectively. Next, using mixture randomization we get
\begin{equation*}
  \begin{aligned}
    \Pr \left(\sum_{k=1}^{n} U_k \leq t, \sum_{k=1}^{n} U_k + E_1 > t \right)
    &= \Pr \left(\sum_{k=1}^{n} U_k \leq t\right) - \Pr \left(\sum_{k=1}^{n} U_k + E_1 \leq t\right)\\
    &= \sum_{k=0}^{n} H(t; (k, \beta_1), (n-k, \beta_2))
     {n\choose k} p^k (1-p)^{n-k},
  \end{aligned}
\end{equation*}
where $H(x; \alpha_1, \beta_1, \alpha_2, \beta_2) = F(x; \alpha_1,
\beta_1, \alpha_2, \beta_2) - F(x; \alpha_1 + 1, \beta_1, \alpha_2, \beta_2)$.

Similarly,
\begin{equation*}
  \begin{aligned}
    \Pr \left(\sum_{k=1}^{n} U_k \leq t, \sum_{k=1}^{n} U_k + E_2 > t\right)
    &= \Pr \left(\sum_{k=1}^{n} U_k \leq t\right) - \Pr \left(\sum_{k=1}^{n} U_k + E_2 \leq t\right)\\
    &= \sum_{k=0}^{n} H(t; (n-k, \beta_2), (k, \beta_1))
     {n\choose k} p^k (1-p)^{n-k}.
  \end{aligned}
\end{equation*}
That is
\begin{equation}\label{eq:raw}
  \begin{aligned}
    &\Pr (N(t) = n)\\
    =& \sum_{k=0}^{n} \left[p  H(t; (k, \beta_1), (n-k, \beta_2))
      + (1-p)  H(t; (n-k, \beta_2), (k, \beta_1))\right] {n\choose k} p^k (1-p)^{n-k}
  \end{aligned}
\end{equation}

Now it is sufficient to show that for any positive integers $\alpha_1$,
$\alpha_2$ and $y, \beta_1, \beta_2 > 0$,
\begin{equation}
  \label{eq:recur}
  H(y; (\alpha_1, \beta_1), (\alpha_2, \beta_2))
  = \frac{y^{\alpha_1 + \alpha_2} e^{-y/\beta_1}}
  {\beta_1^{\alpha_1} \beta_2^{\alpha_2} \Gamma(\alpha_1 + \alpha_2 + 1)}
  {}_1F_1(\alpha_2; \alpha_1 + \alpha_2 + 1; y(1/\beta_1 - 1/\beta_2)).
\end{equation}
Firstly, from equation~\eqref{eq:pco2ga} we have
\begin{align*}
  H(y; (\alpha_1, \beta_1), (\alpha_2, \beta_2))
  = & \left(\frac{\beta_1}{\beta_2}\right)^{\alpha_2}
      \sum_{k = 0}^{\infty} \binom{\alpha_2 + k - 1}{k} \left(1-\beta_1/\beta_2\right)^k
      G(y/\beta_1; k + \alpha_1 + \alpha_2)\\
    &- \left(\frac{\beta_1}{\beta_2}\right)^{\alpha_2}
      \sum_{k = 0}^{\infty} \binom{\alpha_2 + k - 1}{k} \left(1-\beta_1/\beta_2\right)^k
      G(y/\beta_1; k + \alpha_1 + \alpha_2 + 1)\\
  = & \left(\frac{\beta_1}{\beta_2}\right)^{\alpha_2}
      \sum_{k = 0}^{\infty} \binom{\alpha_2 + k - 1}{k} (1-\beta_1/\beta_2)^k
      \frac{(y/\beta_1)^{k+\alpha_1+\alpha_2} e^{-y/\beta_1}}
      {\Gamma(k + \alpha_1 + \alpha_2 + 1)}\\
  = & \left(\frac{\beta_1}{\beta_2}\right)^{\alpha_2}
    \left(\frac{y}{\beta_1}\right)^{\alpha_1 + \alpha_2}
    e^{-y/\beta_1}
    \sum_{k=0}^{\infty}
    \frac{\binom{\alpha_2 + k - 1}{k} \left[y(1/\beta_1 - 1/\beta_2)\right]^k}
         {\Gamma(k + \alpha_1 + \alpha_2 + 1)},
\end{align*}
where the second equation follows from
\[
G(y; \alpha) - G(y; \alpha + 1) = \frac{y^\alpha e^{-y}}{\Gamma(\alpha + 1)}.
\]
Because of the identities
\[
\binom{\alpha_2 + k + 1}{k} = \frac{(\alpha_2)_k}{k!}
\]
and
\[
  \Gamma(k + \alpha_1 + \alpha_2 + 1) = (\alpha_1 + \alpha_2 + 1)_k
  \Gamma(\alpha_1 + \alpha_2 + 1),
\]
we finally obtain
\begin{align*}
  H(y; (\alpha_1, \beta_1), (\alpha_2, \beta_2))
    = & \frac{y^{\alpha_1 + \alpha_2} e^{-y/\beta_1}}
    {\beta_1^{\alpha_1} \beta_2^{\alpha_2} \Gamma(\alpha_1 + \alpha_2 + 1)}
    \sum_{k=0}^{\infty} \frac{(\alpha_2)_k [y(1/\beta_1 - 1/\beta_2)]^k}
    {(\alpha_1 + \alpha_2 + 1)_k k!}\\
    = & \frac{y^{\alpha_1 + \alpha_2} e^{-y/\beta_1}}
    {\beta_1^{\alpha_1} \beta_2^{\alpha_2} \Gamma(\alpha_1 + \alpha_2 + 1)}
    {}_1F_1(\alpha_2; \alpha_1 + \alpha_2 + 1; y(1/\beta_1 - 1/\beta_2)).
\end{align*}
% The proof is completed by plugging~\eqref{eq:recur} into~\eqref{eq:raw}.
\end{proof}

\begin{Corollary}\label{cor:renew}
  Let $\{ U_k \}_{k \geq 1}$ be independent and identically distributed random
  variables distributed as a mixture of $S$ exponential distributions,
  $\mathcal{E}xp(\beta_s)$ with weight $p_s$, $s=1, \dots, S$, and
  $\sum_{s=1}^{S} p_s = 1$. For $n \geq 0$,
  \begin{align*}
    & \Pr\big(N(t) = n\big) \\
    =& \sum_{k_1=0}^{n} \sum_{k_2=0}^{n-k_1} \dots
      \sum_{k_{S-1}=0}^{n-k_1- \dots -k_{S-2}} \left[ \left( \sum_{s=1}^{S} p_s
      H_s(t; (k_1, \beta_1), \dots, (k_S, \beta_S)) \right) \frac{n!}{k_1!k_2! \dots k_S!}
    p_1^{k_1} \dots p_S^{k_S} \right],
  \end{align*}
  where
  \begin{align*}
    H_s(t; (k_1, \beta_1), \dots, (k_S, \beta_S))
    =\, & F_S(t; (k_1, \beta_1), \dots, (k_s, \beta_s), \dots, (k_S, \beta_S)) \\
      & - F_S(t; (k_1, \beta_1), \dots, (k_s + 1, \beta_s), \dots, (k_S, \beta_S)),
  \end{align*}
  $s = 1, \ldots, S$,
  $k_S = n-k_1- \dots - k_{S-1}$, and
  $F_S$ is the distribution function of the convolution of $S$ independent
  gamma variables with shape parameter $(k_1, \dots, k_S)$ and scale
  parameter $(\beta_1, \dots, \beta_S)$.
\end{Corollary}

\begin{proof}
  Note that
  \begin{equation*}
    \Pr(N(t) = n)
    = \Pr \left(\sum_{k=1}^{n} U_k \leq t,
      \sum_{k=1}^{n+1} U_k > t \right)
    = \sum_{s=1}^S p_s \Pr \left(\sum_{k=1}^{n} U_k \leq t,
      \sum_{k=1}^{n} U_k + E_s > t\right),
  \end{equation*}
  where $E_s$ are independent of $\{ U_k \}_{k \geq 1}$ with
  an $\mathcal{E}xp(\beta_s)$ distribution and
  \begin{equation*}
    \begin{aligned}
      &\Pr \left(\sum_{k=1}^{n} U_k \leq t, \sum_{k=1}^{n} U_k + E_s > t\right)
      = \Pr \left( \sum_{k=1}^{n} U_k \leq t \right)
      - \Pr \left( \sum_{k=1}^{n} U_k + E_s \leq t \right)\\
      =& \sum_{k_1=0}^{n} \sum_{k_2=0}^{n-k_1} \dots \sum_{k_{S-1}=0}^{n-k_1- \dots -k_{S-2}}
      H_s(t, (k_1, \beta_1), \dots, (k_S, \beta_S)) \frac{n!}{k_1!k_2! \dots k_S!}
    p_1^{k_1} \dots p_S^{k_S}.
    \end{aligned}
  \end{equation*}
  By substitution, we can complete proof.
\end{proof}

\begin{table}[tbp]
  \caption{Timing comparison (in microseconds) using different methods in $H$
    to evaluate $\Pr(N(10) = n)$ in the renewal process application when $S = 2$.}
  \label{tab:rnew2}
\begin{center}
  \begin{tabular}{rrrrrr}
    \toprule
    \multicolumn{3}{c}{Parameters} & \multicolumn{3}{c}{Time Informations}\\
    \cmidrule(r){1-3} \cmidrule(r){4-6}
    $\beta_1$ & $\beta_2$ & $n$ & Mathai & Moschopoulos & Proposed\\
    \midrule
  0.4 & 0.3 & 27 & 4,993 & 36,314 & 33 \\ 
   &  & 32 & 5,788 & 43,129 & 34 \\ 
   &  & 40 & 6,981 & 54,637 & 41 \\ 
  4 & 0.3 & 10 & 2,675 & 23,613 & 15 \\ 
   &  & 18 & 4,667 & 42,529 & 23 \\ 
   &  & 30 & 7,525 & 65,408 & 58 \\ 
  4 & 3 & 2 & 308 & 951 & 9 \\ 
   &  & 3 & 380 & 1,497 & 15 \\ 
   &  & 5 & 556 & 3,226 & 16 \\ 
    \bottomrule
  \end{tabular}
\end{center}
\end{table}

To demonstrate the computational efficiency of the result of
Proposition~\ref{thm:renew}, we performed a numerical study for $S = 2$.
The scale parameters of the exponential distributions were set to be
$(\beta_1, \beta_2) \in \{(0.4, 0.3), (4, 0.3), (4, 3)\}$.
For each configuration, we evaluated $\Pr(N(t) = n\}$ for $t = 10$ and
three $n$~values corresponding to the 20th, 50th, and 90th percentiles
of a random sample of size 1,000 from the distribution of $N(10)$.
The $H$ function was evaluated with Mathai's method, Moschopoulos'
method and the proposed result in Proposition~\ref{thm:renew}.
The numerical results are identical and the median timing results from
100 replicates are summarized in Table~\ref{tab:rnew2}.
The method in Proposition~\ref{thm:renew}, which bypasses evaluating two
distribution function of the convolution of two exponential variables, shows a
huge advantage, speeding up Mathai's method by a factor of about 60--200.

\begin{table}[tbp]
  \caption{Timing comparison (in milliseconds) using different methods in $H$
    to evaluate $\Pr(N(10) = n)$ in the renewal process application when $S =3$.}
  \label{tab:rnew3}
\begin{center}
  \begin{tabular}{rrrrrrrrr}
    \toprule
    \multicolumn{4}{c}{Parameters} & \multicolumn{3}{c}{Time Informations}\\
    \cmidrule(r){1-4} \cmidrule(r){5-7}
    $\beta_1$ & $\beta_2$ & $\beta_3$ & $n$ & Mathai & Mosch. & Approx. & Exact Value & Relative Error\\
    \midrule
  0.4 & 0.3 & 0.2 & 36 & 95,440 & 3,033 & 669 & 4.2456e-02 & 3.3887e-03 \\ 
   &  &  & 42 & 132,284 & 4,097 & 895 & 5.7594e-02 & -2.5825e-03 \\ 
   &  &  & 51 & 189,053 & 5,901 & 1,285 & 2.2793e-02 & -4.7991e-04 \\ 
  4 & 0.3 & 0.2 & 10 & 8,729 & 376 & 63 & 2.8303e-02 & 3.6682e-03 \\ 
   &  &  & 19 & 33,312 & 1,249 & 218 & 3.3972e-02 & -2.9601e-05 \\ 
   &  &  & 35 & 113,068 & 3,906 & 674 & 1.4896e-02 & -1.0625e-02 \\ 
  4 & 3 & 0.2 & 5 & 2,791 & 94 & 15 & 5.8889e-02 & 3.3020e-04 \\ 
   &  &  & 10 & 12,937 & 382 & 62 & 6.2835e-02 & -2.0693e-03 \\ 
   &  &  & 19 & 49,028 & 1,229 & 203 & 2.1189e-02 & 1.5217e-03 \\ 
  4 & 3 & 2 & 2 & 31 & 6 & 1 & 1.2854e-01 & 1.2242e-03 \\ 
   &  &  & 4 & 233 & 26 & 3 & 1.8740e-01 & -1.9957e-03 \\ 
   &  &  & 7 & 829 & 62 & 11 & 7.2131e-02 & 1.9813e-03 \\ 
    \bottomrule
  \end{tabular}
\end{center}
\end{table}

For $S > 2$, we performed a similar study with Mathai's and Moschopoulos'
method for evaluating $H_3$ using the result in Corollary~\ref{cor:renew}.
For comparison in accuracy and speed, the approximation method of
\citet{barnabani2017approximation} in evaluating $H_3$ was also included.
The scale parameters of the exponential distributions were set to be
\[
(\beta_1, \beta_2, \beta_3) \in \{ (0.4, 0.3, 0.2), \;
(4, 0.3, 0.2), \; (4, 3, 0.2), \; (4, 3, 2)\}.
\]
Again, for each configuration, we evaluated $\Pr(N(10) = n)$ for three $n$
values corresponding to the 20th, 50th, and 90th percentiles of a random
sample of size 1000 from the distribution of $N(10)$.
The median timing results from 100 replicates and the relative error of the
approximation method are summarized in Table~\ref{tab:rnew3}.
Similar to the comparison reported in Section~\ref{sec:exact}, Moschopoulos'
method is over 10~times faster than Mathai's method in all the cases.
A small sumation of $\beta$'s requires longer computaion time.
The approximation method is over 6~times faster than Moschopoulo's method,
with relative error less than 1~percent.

\section{Conclusions}
\label{sec:con}

We reviewed two exact methods and one approximation method for evaluating
the density and distribution of convolutions of independent gamma variables.
From our study, the method of \citet{mathai1982storage} is the fastest in the
case of $n = 2$ because of efficient GSL implementation of the univariate
Kummer confluent hypergeometric function; when $n > 3$, the method of
\citet{moschopoulos1985distribution} is faster than Mathai's method.
The fast approximation method \citet{barnabani2017approximation} is quite
accurate, which provides a useful tool for applications with $n > 3$.
Implementations of the reviewed methods are available in R package
\texttt{coga} \citep{Rpkg:coga}, which is built on C++ code for fast speed.

The result in Proposition~\ref{thm:renew} for the distribution of the
event count in the renew process application with holding time following a
mixture of exponential distribution is an extremely fast evaluation.
One side-benefit of the proposition is formula~\eqref{eq:recur} for
the difference of distribution functions of convolution of two independent
Erlang distributions with shape parameters that differ by~$1$. It can be
evaluated accurately and efficiently using the GSL implementation of the
confluent hypergeometric function ${}_1F_1$. This difference plays an
important role in computation of distribution of the occupation times for a
certain continuous time Markov chain \citep{pozdnyakov2017discretely},
the defective density in generalized integrated telegraph processes
\citep{zacks2004generalized}, and the first-exit time in a compound process
\citep{perry1999firstexit}.

\singlespacing
\bibliographystyle{Chicago}
\bibliography{coga}

\begin{thebibliography}{}

\bibitem[\protect\citeauthoryear{Abramowitz and Stegun}{Abramowitz and
  Stegun}{1972}]{abramowitz1972handbook}
Abramowitz, M. and I.~A. Stegun (1972).
\newblock {\em Handbook of Mathematical Functions with Formulas, Graphs, and
  Mathematical Tables}, Volume~9.
\newblock Dover, New York.

\bibitem[\protect\citeauthoryear{Akkouchi}{Akkouchi}{2005}]{akkouchi2005convolution}
Akkouchi, M. (2005).
\newblock On the convolution of gamma distributions.
\newblock {\em Soochow Journal of Mathematics\/}~{\em 31\/}(2), 205--211.

\bibitem[\protect\citeauthoryear{Barnabani}{Barnabani}{2017}]{barnabani2017approximation}
Barnabani, M. (2017).
\newblock An approximation to the convolution of gamma distributions.
\newblock {\em Communications in Statistics - Simulation and
  Computation\/}~{\em 46\/}(1), 331--343.

\bibitem[\protect\citeauthoryear{Di~Salvo}{Di~Salvo}{2006}]{di2006exact}
Di~Salvo, F. (2006).
\newblock The exact distribution of the weighted convolution of two gamma
  distributions.

\bibitem[\protect\citeauthoryear{Galassi, Davies, Theiler, Gough, and
  Jungman}{Galassi et~al.}{2009}]{galassi2009gnu}
Galassi, M., J.~Davies, J.~Theiler, B.~Gough, and G.~Jungman (2009).
\newblock {\em GNU Scientific Library: Reference Manual\/} (3 ed.).
\newblock Network Theory Ltd.

\bibitem[\protect\citeauthoryear{Hu, Pozdnyakov, and Yan}{Hu
  et~al.}{2017}]{Rpkg:coga}
Hu, C., V.~Pozdnyakov, and J.~Yan (2017).
\newblock {\em {coga}: Convolution of Gamma Distributions}.
\newblock R package version 0.2.1.

\bibitem[\protect\citeauthoryear{Jain and Consul}{Jain and
  Consul}{1971}]{jain1971generalized}
Jain, G.~C. and P.~C. Consul (1971).
\newblock A generalized negative binomial distribution.
\newblock {\em SIAM Journal on Applied Mathematics\/}~{\em 21\/}(4), 501--513.

\bibitem[\protect\citeauthoryear{Jasiulewicz and Kordecki}{Jasiulewicz and
  Kordecki}{2003}]{jasiulewicz2003convolutions}
Jasiulewicz, H. and W.~Kordecki (2003).
\newblock Convolutions of {E}rlang and of {p}ascal distributions with
  applications to reliability.
\newblock {\em Demonstratio Mathematica. Warsaw Technical University Institute
  of Mathematics\/}~{\em 36\/}(1), 231--238.

\bibitem[\protect\citeauthoryear{Kadri, Smaili, and Kadry}{Kadri
  et~al.}{2015}]{kadri2015markov}
Kadri, T., K.~Smaili, and S.~Kadry (2015).
\newblock {M}arkov modeling for reliability analysis using hypoexponential
  distribution.
\newblock In S.~Kadry and A.~El~Hami (Eds.), {\em Numerical Methods for
  Reliability and Safety Assessment: Multiscale and Multiphysics Systems}, pp.\
   599--620. Cham: Springer.

\bibitem[\protect\citeauthoryear{Mathai}{Mathai}{1982}]{mathai1982storage}
Mathai, A.~M. (1982).
\newblock Storage capacity of a dam with gamma type inputs.
\newblock {\em Annals of the Institute of Statistical Mathematics\/}~{\em
  34\/}(1), 591--597.

\bibitem[\protect\citeauthoryear{Mathai and Saxena}{Mathai and
  Saxena}{1978}]{mathai1978hfunction}
Mathai, A.~M. and R.~K. Saxena (1978).
\newblock {\em The H function with Applications in Statistics and Other
  Disciplines}.
\newblock New Delhi: Wiley.

\bibitem[\protect\citeauthoryear{Moschopoulos}{Moschopoulos}{1985}]{moschopoulos1985distribution}
Moschopoulos, P.~G. (1985).
\newblock The distribution of the sum of independent gamma random variables.
\newblock {\em Annals of the Institute of Statistical Mathematics\/}~{\em
  37\/}(1), 541--544.

\bibitem[\protect\citeauthoryear{Perry, Stadje, and Zacks}{Perry
  et~al.}{1999}]{perry1999firstexit}
Perry, D., W.~Stadje, and S.~Zacks (1999).
\newblock First-exit times for increasing compound processes.
\newblock {\em Communications in Statistics. Stochastic Models\/}~{\em
  15\/}(5), 977--992.

\bibitem[\protect\citeauthoryear{Pozdnyakov, Hu, Meyer, and Yan}{Pozdnyakov
  et~al.}{2017}]{pozdnyakov2017discretely}
Pozdnyakov, V., C.~Hu, T.~Meyer, and J.~Yan (2017).
\newblock On discretely observed {B}rownian motion governed by a continuous
  time {M}arkov chain.
\newblock Technical report, University of Connecticut, Department of
  Statistics.

\bibitem[\protect\citeauthoryear{{R Core Team}}{{R Core Team}}{2017}]{R}
{R Core Team} (2017).
\newblock {\em R: A Language and Environment for Statistical Computing}.
\newblock Vienna, Austria: R Foundation for Statistical Computing.

\bibitem[\protect\citeauthoryear{Sen and Balakrishnan}{Sen and
  Balakrishnan}{1999}]{sen1999convolution}
Sen, A. and N.~Balakrishnan (1999).
\newblock Convolution of geometrics and a reliability problem.
\newblock {\em Statistics \& Probability Letters\/}~{\em 43\/}(4), 421--426.

\bibitem[\protect\citeauthoryear{Sim}{Sim}{1992}]{sim1992point}
Sim, C. (1992).
\newblock Point processes with correlated gamma interarrival times.
\newblock {\em Statistics \& Probability Letters\/}~{\em 15\/}(2), 135--141.

\bibitem[\protect\citeauthoryear{Vellaisamy and Upadhye}{Vellaisamy and
  Upadhye}{2009}]{vellaisamy2009sums}
Vellaisamy, P. and N.~Upadhye (2009).
\newblock On the sums of compound negative binomial and gamma random variables.
\newblock {\em Journal of Applied Probability\/}~{\em 46\/}(1), 272--283.

\bibitem[\protect\citeauthoryear{Zacks}{Zacks}{2004}]{zacks2004generalized}
Zacks, S. (2004).
\newblock Generalized integrated telegraph processes and the distribution of
  related stopping times.
\newblock {\em Journal of Applied Probability\/}~{\em 41\/}(2), 497--507.

\end{thebibliography}

\end{document}